\documentclass[twoside]{JHEP3}
\usepackage{amssymb}
\usepackage{amstext}
\usepackage{amsmath}
\usepackage{mathrsfs}
\usepackage{latexsym}
\usepackage{verbatim}
\usepackage{stmaryrd}
\newcommand{\triadj}{\blacktriangleright}
\newcommand{\tg}{\hat{t}}
\newcommand{\xg}{\hat{x}}
\newcommand{\bolo}[1]{\left\llbracket   #1  \right\rrbracket}
\newlength{\enviropost}
\newcommand{\be}{\begin{equation}}
\newcommand{\ee}{\end{equation}}
\newcommand{\ble}[1]{\begin{equation} \label{#1}}
\newcommand{\bae}{\begin{eqnarray}}
\newcommand{\eae}{\end{eqnarray}}
\newcommand{\fle}[2]%
{\vspace{1.5ex}
\be
\label{#1}
\mbox{%
\setlength{\fboxsep}{3ex}%
\framebox{$\dss #2 $}}
\ee} 
\newcommand{\flec}[2]%
{\vspace{1.5ex}
\be
\label{#1}
\mbox{%
\setlength{\fboxsep}{3ex}%
\framebox{$\dss #2 $}}
\, \, \,  ,
\ee} 
\newcommand{\flep}[2]%
{\vspace{1.5ex}
\be
\label{#1}
\mbox{%
\setlength{\fboxsep}{3ex}%
\framebox{$\dss #2 $}}
\, \, \, .
\ee} 
\newcommand{\equalref}[1]{\stackrel{\rule[-1ex]{0mm}{0mm}%
{}_{(\ref{#1})}}{=}}

\newcommand{\nn}{\nonumber}
\newcommand{\ff}{\nn \\}
\newcommand{\fe}{& = &}

\newtheorem{prop}{Proposition}
\newtheorem{corollary}{Corollary}

\newtheorem{state}{S$\! \!$}
\newtheorem{defin}{D$\! \!$}
\newtheorem{exatitle}{Example}
\newtheorem{problemdef}{Problem}
\newtheorem{soldef}{Solution}
{\textsc{Definition:} \ }% 
{\vspace{\enviropost} \\}
\newenvironment{proof}%
{\noindent \textsc{Proof}:\ }% 
{\hfill $\blacksquare$  \vspace{\enviropost} \\}
{\begin{soldef} 
\label{#2:Sol} 
#1
\marginpar{(\textbf{\pageref{#2:Pro}})}
\end{soldef}}% 
{\hfill $\Box$ \vspace{.5\enviropost} \\}
{\begin{exatitle} \label{#2} #1 \end{exatitle}}% 
{\hfill $\Box$ \vspace{\enviropost} \\}
\newenvironment{myexample}[2]%
{\begin{exatitle} \label{#2} #1 \end{exatitle}}%
{\hfill $\Box$ \\}
{\begin{problemdef} 
\label{#2:Pro} 
#1 
\marginpar{(\textbf{\pageref{#2:Sol}})}
\end{problemdef}}% 
{\hfill  \vspace{.5\enviropost} \\}
{\noindent {\bf Aside:} \small}% 
{\hfill \rule[-3mm]{0mm}{0mm}$\Diamond$\\}
{%
\vspace{3mm}  
\begin{center}
\begin{minipage}{.8\textwidth} 
\begin{state} 
\label{#1}%
}% 
{%
\end{state}
\end{minipage}
\end{center}
\vspace{3mm}
}
{%
\vspace{3mm}  
\begin{center}
\begin{minipage}{.8\textwidth} 
\begin{defin} 
\label{#1}%
}% 
{%
\end{defin}
\end{minipage}
\end{center}
\vspace{3mm}
}
\newcommand{\mvb}{\mathversion{bold}}
\newcommand{\mvn}{\mathversion{normal}}
\newcommand{\dss}{\displaystyle}
\newcommand{\id}{\mathop{\rm id}}

\newcommand{\ip}[2]{\left\langle #1, #2\right\rangle}

\newcommand{\ket}[1]{| #1 \rangle}

\newcommand{\eg}{\hbox{\em e.g.{}}}
\newcommand{\etc}{\hbox{\em etc.{}}}
\newcommand{\ie}{\hbox{\em i.e.{}}}

\newcommand{\rhs}{\hbox{r.h.s.{}}}

\newcommand{\calA}{\mathcal{A}}

\newcommand{\calF}{\mathcal{F}}

\newcommand{\calO}{\mathcal{O}}
\newcommand{\calP}{\mathcal{P}}

\newcommand{\calR}{\mathcal{R}}

\newcommand{\calU}{\mathcal{U}}

\title{%
Star Product and Invariant Integration for Lie type\\[1mm] Noncommutative Spacetimes%
}
\author{%
Chryssomalis Chryssomalakos\thanks{%
On sabbatical leave from the Instituto de Ciencias Nucleares, Universidad Nacional Aut\'onoma de M\'exico, Apdo. Postal 70-543, 04510 M\'exico, D.F., M\'EXICO.}\\
School of Applied Sciences\\
Department of Physics and Applied Mathematics\\
National Technical University\\
Zografou Campus, 15780 Athens, Greece\\
E-mail: \email{chryss@nucleares.unam.mx}}
\author{%
Elias Okon\\
Instituto de Ciencias Nucleares\\
Universidad Nacional Aut\'onoma de M\'exico\\
Apdo. Postal 70-543, 04510 M\'exico, D.F., M\'EXICO\\
E-mail: \email{eliokon@nucleares.unam.mx}}
\abstract{%
We present a star product for noncommutative spaces of Lie type, including the so called ``canonical'' case by introducing a central generator, which is compatible with translations and admits a simple, manageable definition of an invariant integral. A quasi-cyclicity property for the latter is shown to hold, which reduces to exact cyclicity when the adjoint representation of the underlying Lie algebra is traceless. Several explicit examples illuminate the formalism, dealing with $\kappa$-Minkowski spacetime and the Heisenberg algebra (``canonical'' noncommutative 2-plane).
}
\begin{document}
%%%%%%%%%%%%%%%%%%%%%%%%%%%%%%%%%%%%%%%%%%%%%%%%%%%%%%%%%%%%%%
%%%%%%%%%%%%%%%%%%%%%%%%%%%%%%%%%%%%%%%%%%%%%%%%%%%%%%%%%%%%%%
%%%%%%%%%%%%%%%%%%%%%%%%%%%%%%%%%%%%%%%%%%%%%%%%%%%%%%%%%%%%%%
% Titlepage
%%%%%%%%%%%%%%%%%%%%%%%%%%%%%%%%%%%%%%%%%%%%%%%%%%%%%%%%%%%%%%
%%%%%%%%%%%%%%%%%%%%%%%%%%%%%%%%%%%%%%%%%%%%%%%%%%%%%%%%%%%%%%
%%%%%%%%%%%%%%%%%%%%%%%%%%%%%%%%%%%%%%%%%%%%%%%%%%%%%%%%%%%%%%
\section{Introduction}
\label{Intro}
%%%%%%%%%%%%%%%%%%%%%%%%%%%%%%%%%%%%%%%%%%%%%%%%%%%%%%%%%%%%%%
%%%%%%%%%%%%%%%%%%%%%%%%%%%%%%%%%%%%%%%%%%%%%%%%%%%%%%%%%%%%%%
The motivation behind noncommutative versions of quantum field theory is by now adequately expounded in the literature. The massive output,
in this direction, of the physics-oriented community in the last couple of decades has involved in a substantial way star products, in bewildering numbers, shapes and forms. The present paper adds one more such product in the list, suitable for Lie-type noncommutative spacetimes, and proves several useful properties, using all along explicit examples to clarify the general theorems. What we perceive as virtues for the proposed star product are the following:
\begin{enumerate}
\item 
Its definition is natural and easy to compute, eliminating the need for the plethora of {\em ad hoc} ordering prescriptions other definitions rely on.
\item
It coexists harmoniously with translations, in the sense that the corresponding $\Omega$ map commutes with them.
\item
It admits a simple, manageable definition of a translation-invariant integral --- we compute several such integrals explicitly in various examples.
\item
The above integral is shown to also satisfy the trace property (cyclicity) for a certain class of underlying Lie algebras, while, in the general case, a simple generalization of the trace property is shown to hold.
\end{enumerate}
On the other hand, the list of truly undesirable features of our star product is, in our mind, currently void. Given this happy state of affairs, we believe it is worth presenting in some detail, in the hope that it might prove useful in physical applications. 

As mentioned already, we deal with noncommutative spacetimes of the Lie type, namely, those for which the commutators of the coordinate functions are linear in those same functions. We also adopt a physical point of view, meaning that our reserves of tolerance for the legitimacy of our tools will be considerable, provided they promise to be useful. A natural playground for what is considered here is the $\kappa$-Minkowski spacetime, but the also much studied ``canonical'' noncommutative spacetime fits in our scheme, with minimal modifications --- these two constructs form the backdrop of our examples. 

The subject we touch upon here is well studied, with several hundred papers published in the general area of deformation quantization during the last three decades. Accordingly, we limit our list of references to only those works that bear a direct relation to ours (and that we know of): 
 Refs.~\cite{Dimitrijevic:2003wv,Dimitrijevic:2003pn,Madore:2000en,%
kosinski-2000-50,Aschieri:2006ye,Szabo:2001kg} deal with quantum field theory on noncommutative spacetimes, \cite{Daszkiewicz:2007eq} with noncommutative translations, \cite{Moller:2004sk} with symmetry-invariant integration on $\kappa$-Minkowski, \cite{Zachos:1999mp,kosinski-2001-64,graciabondia-2002-0204,sykora-2004,hirshfeld-2002-298} with more general  star product matters and \cite{Wess:2006cm,Sitarz:1994rh,Majid:1994cy} with noncommutative differential calculi. On the more mathematical side,  Refs.{}~\cite{gutt:1983a,bordemann.neumaier.waldmann:1998a} deal with  star products on the cotangent bundle of a Lie group, their relation to the Kontsevich~\cite{kontsevich:2003a}  star product being explored in~\cite{dito:1999a}. This last reference, along with several others, was brought to our attention by Stefan Waldmann after this paper had essentially been written up --- its Lemma 2 establishes the equivalence of ours and Gutt's  star product, although not its explicit form. Cyclic integrals (trace functionals) for  star products are discussed in~\cite{bieliavsky.bordemann.gutt.waldmann:2003a}, while KMS functionals for symplectic manifolds are dealt with in~\cite{basart.flato.lichnerowicz.sternheimer:1984a,basart.lichnerowicz:1985a,bordemann.roemer.waldmann:1998a}. The quasi-cyclicity property of our integral seems to be new. Further references to the field can be found in the very extensive list at \texttt{http://idefix.physik.uni-freiburg.de/$\sim$star/}.

With so much already said on the subject, it is perhaps inevitable that our construction coincides, in particular cases, with choices already made by other authors. Even then though, our general theorems guarantee properties that may have passed unnoticed, and put various approaches in the physics literature under the same, conceptually  unifying, roof.

Regarding the structure of the paper, a glance at the table of contents should give an idea of its organization. We point out that some of our proofs make use of Hopf algebra techniques that might not be familiar to some of the readers. We have opted to present this material separately in an appendix, so as to not disrupt the flow of our presentation.
%%%%%%%%%%%%%%%%%%%%%%%%%%%%%%%%%%%%%%%%%%%%%%%%%%%%%%%%%%%%%%
%%%%%%%%%%%%%%%%%%%%%%%%%%%%%%%%%%%%%%%%%%%%%%%%%%%%%%%%%%%%%%
\section{The Star Product and Some of Its Properties}
\label{Acspasoip}
%%%%%%%%%%%%%%%%%%%%%%%%%%%%%%%%%%%%%%%%%%%%%%%%%%%%%%%%%%%%%%
%%%%%%%%%%%%%%%%%%%%%%%%%%%%%%%%%%%%%%%%%%%%%%%%%%%%%%%%%%%%%%
%%%%%%%%%%%%%%%%%%%%%%%%%%%%%%%%%%%%%%%%%%%%%%%%%%%%%%%%%%%%%%
\subsection{Definition of the star product}
\label{Dotsp}
%%%%%%%%%%%%%%%%%%%%%%%%%%%%%%%%%%%%%%%%%%%%%%%%%%%%%%%%%%%%%%
Consider the $N$-dimensional Lie algebra $\mathfrak{g}$ with generators $\hat{z}_A$, $A=1,\ldots,N$, and commutation relations 
\be
\label{zzcrs}
[\hat{z}_A, \hat{z}_B]= \lambda {f_{AB}}^C \hat{z}_C
\, ,
\ee
where $\lambda$ is a formal parameter.
We will denote the left invariant vector fields (LIVF) corresponding to $\hat{z}_A$ by $Z_A$. Expressed in terms of coordinates $\zeta^A$ on the corresponding group manifold and the dual partials $\partial_A \equiv \partial/\partial \zeta^A$, they are given by 
\be
\label{Zgs}
Z_A =  \psi_A^{\phantom{A}B} (\zeta) \, \partial_B
\, .
\ee
Evaluated at the identity of the group, $Z_A$ reduces to the partial $\partial_A$, which implies 
\be
\label{psicounit}
\epsilon(\psi_A^{\phantom{A}B})=\delta_A^{\phantom{A}B}
\, ,
\ee
where the {\em counit} $\epsilon(f(\zeta))$ of a function $f(\zeta)$, is the value of the function at the identity of the group. The $\psi_A^{\phantom{A}B}$ also depend on $\lambda$. In the limit $\lambda \rightarrow 0$ of an abelian algebra, $Z_A$ also reduces to the partial $\partial_A$, so that
\be 
\label{limZl}
\lim_{\lambda \rightarrow 0} \psi_A^{\phantom{A}B}(\zeta)=\delta_A^{\phantom{A}B}
\, .
\ee
We use throughout coordinates that vanish at the identity, so that $\epsilon(f(\zeta))=f(0)$. In the sequel, it will prove convenient to use the notation $z_A \equiv \partial/\partial \zeta^A$ --- notice that the $\hat{z}_A$ satisfy the Lie algebra relations, Eq.~(\ref{zzcrs}), while the $z_A$ commute among themselves, like partials do. In this notation,
\be
\label{zzetacrs}
z_A \zeta^B=\delta_A^{\phantom{A}B} + \zeta^B z_A
\, ,
\ee
and
\be 
\label{Zgs2}
Z_A = \psi_A^{\phantom{A}B}(\zeta) z_B
\, .
\ee
We denote by $\tilde{\calU}_{\mathfrak{g}}[[\lambda]] \equiv \hat{\calF}$ the algebra of formal power series in $\lambda$, with coefficients in the (suitably completed) universal enveloping algebra of $\mathfrak{g}$, and think of it as the (noncommutative) algebra of functions on a quantum space. Similarly, $\calF$ will denote formal power series in $\lambda$, with coefficients in the algebra of $C^\infty$ functions on $\mathbb{R}^N$. Our first task is to provide a pair $(\Omega,\star)$, where 
\be
\label{Omegadefgen}
\Omega \colon \hat{\calF} \rightarrow \calF
\ee
is an invertible linear map, and $\star \colon \calF \otimes \calF \rightarrow \calF$ is a (noncommutative) product between elements of $\calF$, such that
\be 
\label{Oscomp}
\Omega(\hat{f}\hat{g})=\Omega(\hat{f})\star \Omega(\hat{g})
\, .
\ee
In the above equation $\hat{f}$ and $\hat{g}$ denote functions of the $\hat{z}$'s, a convention we use throughout the paper --- similarly, symbols like $f$, $g$ \etc, will denote functions of the $z$'s. At this point we introduce a further simplification in the notation, by putting $\Omega(\hat{f}) \equiv f$, in other words, $\hat{f}$ is to be interpreted as $\Omega^{-1}(f)$, and~(\ref{Oscomp}) becomes
\be 
\label{Oscomp1}
\Omega(\hat{f}\hat{g})=f \star g
\, .
\ee
A final ingredient of our construction is a map $\omega$ that applied to a general function $\hat{F}(\zeta,z)$ returns a function of the $z$'s only, as follows: using the commutation relations~(\ref{zzetacrs}), bring all $\zeta$'s in $\hat{F}$ to the left --- denote the resulting expression 
by 
\be 
\label{notation1}
\sum_i F_i^{(\bar{1})}(\zeta) F_i^{(2)}(z) \equiv F^{(\bar{1})}(\zeta) F^{(2)}(z)=\hat{F}(\zeta,z)
\, .
\ee 
Notice the final equality --- all we have done with $\hat{F}$ is write it out using a particular ordering. It is also important to keep in mind the suppressed summation in the middle form of the above equation. $\omega(\hat{F})$ is now defined by
\be 
\label{omegadef}
\omega(\hat{F}(\zeta,z))=\epsilon\big(F^{(\bar{1})}(\zeta)\big)
F^{(2)}(z)=F^{(\bar{1})}(0)F^{(2)}(z)
\, ,
\ee
\ie, after bringing all $\zeta$'s in $\hat{F}$ to the left, we set them equal to zero --- the resulting function of the $z$'s is $\omega(\hat{F})$.

We may now state our definition of $\Omega \big( \hat{f} (\hat{z}) \big)$:
\be 
\label{Omegadef1}
\Omega\big(\hat{f}(\hat{z})\big)=\omega\big(\hat{f}(Z)\big)
\, .
\ee
In a more verbal mood, finding $\Omega(\hat{f}(\hat{z}))$ entails the following procedure
\begin{enumerate}
\item
Make the substitution $\hat{z}_A \mapsto Z_A = \psi_A^{\phantom{A}B}(\zeta) z_B$ in $\hat{f}(\hat{z})$. Notice that there is no ordering ambiguity involved in this step, given that the $Z_A$ satisfy the same Lie algebra as the $\hat{z}_A$.
\item
In the resulting expression, bring all the $\zeta$'s to the left, using the commutation relations~(\ref{zzetacrs}).
\item
Set $\zeta^A=0$ --- the resulting function of the $z$'s is $\Omega(\hat{f}(\hat{z})) = f(z)$.
\end{enumerate}
The $\star$-product is now {\em defined} by~(\ref{Oscomp1}). It should be clear from the above definition, that both $\Omega$ and $\star$ depend not only on the Lie algebra $\mathfrak{g}$, but also on the basis $\{\hat{z}_A\}$ {\em and} the choice of coordinates on the group manifold. Still, all properties stated in what follows, hold true in general.
%%%%%%%%%%%%%%%%%%%%%%%%%%%%%%%%%%%%%%%%%%%%%%%%%%%%%%%%%%%%%%
\subsection{Properties}
\label{Properties}
%%%%%%%%%%%%%%%%%%%%%%%%%%%%%%%%%%%%%%%%%%%%%%%%%%%%%%%%%%%%%%
%%%%%%%%%%%%%%%%%%%%%%%%%%%%%%%%%%%%%%%
\begin{prop}
\label{Omegafg}
For all $\hat{f}$, $\hat{g}$ it holds
\be 
\label{omegaprop1}
\omega\big(\hat{f}(Z)\hat{g}(Z)\big)
=
\omega(f(z)\hat{g}(Z)\big)
\, .
\ee
\end{prop}
%%%%%%%%%%%%%%%%%%%%%%%%%%%%%%%%%%%%%%%
\begin{proof}
We have
\bae 
\label{prop1proof1}
\omega\big(
\hat{f}(Z)\hat{g}(Z)
\big)
\fe
\omega\big(
f^{(\bar{1})}(\zeta) f^{(2)}(z) \,
g^{(\bar{1})}(\zeta) g^{(2)}(z)
\big)
\ff
\fe
\omega \big(
f^{(\bar{1})}(0) f^{(2)}(z) \,
g^{(\bar{1})}(\zeta) g^{(2)}(z)
\big)
\ff
\fe
\omega(f(z)\hat{g}(Z)\big)
\, .
\eae
In the first equality, we just reordered the $\zeta$'s and $z$'s in $\hat{f}(Z)$ and $\hat{g}(Z)$ individually. To compute $\omega$, one still needs to commute $g^{(\bar{1})}(\zeta)$ to the left, past $f^{(2)}(z)$ and then set $\zeta=0$. However, one may already set $\zeta=0$ in $f^{(\bar{1})}(\zeta)$, since anything with a $\zeta$ to its left will eventually vanish, regardless of internal reorderings. The third equality is by definition.
\end{proof}
Above we have used the notation $\hat{f}(Z)=f^{(\bar{1})}(\zeta) f^{(2)}(z)$, dropping the hat from the $f$'s in the \rhs. This is consistent with~(\ref{notation1}) and the fact that $f^{(\bar{1})}$ and $f^{(2)}$ are both functions of commuting variables. Since $\Omega$ is invertible (see below), no ambiguity arises: given an $f(z)$, one can compute $\hat{f}(\hat{z})$, and then $f^{(\bar{1})}(\zeta) f^{(2)}(z)$ as in~(\ref{notation1}).
%%%%%%%%%%%%%%%%%%%%%%%%%%%%%%%%%%%%%%%
%%%%%%%%%%%%%%%%%%%%%%%%%%%%%%%%%%%%%%%
\begin{prop}
$\Omega$ is invertible.
\end{prop}
%%%%%%%%%%%%%%%%%%%%%%%%%%%%%%%%%%%%%%%
\begin{proof}
It is easy to see that $\Omega(\hat{z}_{A_1}\hat{z}_{A_2}\ldots \hat{z}_{A_n})=
z_{A_1}z_{A_2}\ldots z_{A_n}+ \calO(\lambda)$. Thus, with a suitable ordering of the basis elements in $\calF$, $\hat{\calF}$, the matrix of the map $\Omega$ is, say, upper triangular, with units along the diagonal and, hence, invertible. 
\end{proof}
%%%%%%%%%%%%%%%%%%%%%%%%%%%%%%%%%%%%%%%
%%%%%%%%%%%%%%%%%%%%%%%%%%%%%%%%%%%%%%%
\begin{prop}
In the notation of~(\ref{notation1}), it holds
\be 
\label{fstargprop}
f(z) \star g(z) = 
\left(
g^{(\bar{1})}(\partial_z) \triangleright f(z)
\right)
g^{(2)}(z)
\, ,
\ee
where $\cdot \triangleright \cdot$ denotes the action of a differential operator on a function.
\end{prop}
%%%%%%%%%%%%%%%%%%%%%%%%%%%%%%%%%%%%%%%
\begin{proof}
From~(\ref{zzetacrs}) it can be shown that $f(z)\zeta^A=\partial f(z)/\partial z_A+\zeta^A f(z)$, so that 
\be
\label{omegafzeta}
\omega(f(z)\zeta^A)=\frac{\partial f(z)}{\partial z_A}
\, ,
\ee
and, more generally,
\be 
\label{omegafgzeta}
\omega\big(f(z)g(\zeta)h(z)\big)
=
\big(g(\partial_z)\triangleright f(z) \big) h(z)
\, ,
\ee where $g(\partial_z)$ is obtained from $g(\zeta)$ by the substitution $\zeta^A \mapsto \partial_{z_A}$ (notice that $h(z)$ above can be taken outside of $\omega$). From~(\ref{omegaprop1}) then we get
\bae 
f \star g
\fe
\Omega(\hat{f}(\hat{z})\hat{g}(\hat{z}))
\ff
\fe
\omega 
\left(
f(z) g^{(\bar{1})}(\zeta) g^{(2)}(z)
\right)
\ff
& \equalref{omegafgzeta} &
\left(
g^{(\bar{1})}(\partial_z) \triangleright f(z)
\right) \, g^{(2)}(z)
\, ,
\eae
where we used the equality of left and right actions of partials on functions.
\end{proof}
Notice that the above star product is $\mathfrak{g}$-covariant by construction, \ie, it is bilinear and satisfies $[z_A,z_B]_*=\lambda f_{AB}^{\phantom{AB}C}z_C$. Then, Lemma 2 of~\cite{dito:1999a} implies that it is equivalent to the Gutt star product.
%%%%%%%%%%%%%%%%%%%%%%%%%%%%%%%%%%%%%%%%%%%%%%%%%%%%%%%%%%%%%%
\subsubsection{Commutativity with translations}
\label{Cwt}
%%%%%%%%%%%%%%%%%%%%%%%%%%%%%%%%%%%%%%%%%%%%%%%%%%%%%%%%%%%%%%
The statement of further properties necessitates the introduction of a concept of translation, for both commuting and noncommuting coordinates. In the former case, we formalize the usual $z_A \mapsto z_a +z'_A$ operation via the definition of a {\em coproduct} map $\Delta$,
\be 
\label{Deltadef}
\Delta \colon \calF \rightarrow \calF \otimes \calF 
\, ,
\qquad
z_A \mapsto z_A \otimes 1 + 1 \otimes z_A
\, ,
\ee
extended by linearity and multiplicativity to arbitrary functions in $\calF$ ($\alpha$, $\beta \in \mathbb{R}$),
\be 
\label{lineamulti} 
\Delta(\alpha f+\beta g)=\alpha \Delta(f)+\beta \Delta(g)
\, ,
\qquad
\qquad
\Delta(fg)=\Delta(f)\Delta(g)
\, .
\ee
For noncommuting functions, a similar definition works: the coproduct $\hat{\Delta}$ is given by
\be
\label{hDeltadef}
\hat{\Delta} \colon \hat{\calF} \rightarrow \hat{\calF} \otimes \hat{\calF} 
\, ,
\qquad
\hat{z}_A \mapsto \hat{z}_A \otimes 1 + 1 \otimes \hat{z}_A
\, ,
\ee
extended again by linearity and multiplicativity to arbitrary functions in $\hat{\calF}$. The nontrivial aspect of this definition is that $\hat{\Delta}$ is compatible with the algebra structure, namely,
\be 
\label{hDeltaalg}
[\hat{\Delta}(\hat{z}_A), \, \hat{\Delta}(\hat{z}_B)]
=\lambda f_{AB}^{\phantom{AB}C} \hat{\Delta}(\hat{z}_C)
\, ,
\ee
where the product in $\hat{\calF} \otimes \hat{\calF}$ is given by
$(\hat{f} \otimes \hat{g})(\hat{h} \otimes \hat{r})=\hat{f}\hat{h} \otimes \hat{g}\hat{r}$. As can be appreciated already from~(\ref{Deltadef}), (\ref{hDeltadef}), the coproduct of a function, commutative or not, involves, in general, a sum over tensor products --- the so-called Sweedler notation is standard,
\be 
\label{Sweedler}
\Delta(f)=\sum_i f^i_{(1)} \otimes f_{(2)}^i \equiv 
f_{(1)} \otimes f_{(2)}
\, ,
\ee
and similarly for $\hat{\Delta}(\hat{f})$ (which coproduct is used, $\Delta$ or $\hat{\Delta}$, can be inferred from its argument).
 
We may now state the very important 
%%%%%%%%%%%%%%%%%%%%%%%%%%%%%%%%%%%%%%%
\begin{prop}
\label{OmegaDeltaprop}
$\Omega$ commutes with translations, in the sense that
\be 
\label{OmegaDelta}
\big( \Omega \otimes \Omega \big) \hat{\Delta}(\hat{f})
= 
\Delta \big( \Omega(\hat{f}) \big)
=
\Delta(f)
\quad \text{or, in Sweedler notation,}
\quad
\Omega(\hat{f}_{(1)}) \otimes \Omega(\hat{f}_{(2)})
=
f_{(1)} \otimes f_{(2)}
\, .
\ee
\end{prop}
%%%%%%%%%%%%%%%%%%%%%%%%%%%%%%%%%%%%%%%
The proof is given in the appendix. It is worth emphasizing that the $\star$-product is {\em not} translation invariant, \ie, in general,
\be 
\label{starnotinv}
(f \star g)(z+a) \neq f(z+a) \star g(z+a)
\, .
\ee
This is due to the fact that the $\star$-operation involves explicitly the coordinates. One may use the notation $\star_z$ to denote this explicitly, in which case translation invariance holds in the form
\begin{equation}
\label{starinv}
(f \star_z g)(z+a) = f(z+a) \star_{z+a} g(z+a)
\, .
\end{equation}
%%%%%%%%%%%%%%%%%%%%%%%%%%%%%%%%%%%%%%%
%%%%%%%%%%%%%%%%%%%%%%%%%%%%%%%%%%%%%%%%%%%%%%%%%%%%%%%%%%%%%%
\subsubsection{$\star$-product homomorphism}
\label{Sph}
%%%%%%%%%%%%%%%%%%%%%%%%%%%%%%%%%%%%%%%%%%%%%%%%%%%%%%%%%%%%%%
\begin{prop}
$\Delta$ is an algebra homomorphism of the $\star$-product:
\be
\Delta(f \star g)
=
\Delta(f) \star \Delta(g) 
\equiv 
f_{(1)} \star g_{(1)} \otimes f_{(2)} \star g_{(2)} 
\, .
\ee
\end{prop}
\begin{proof}
\bae
\Delta(f \star g) 
& \equalref{Oscomp1} & 
\Delta \Omega( \hat{f} \, \hat{g}) 
\ff
& \equalref{OmegaDelta} &
\big( \Omega \otimes \Omega \big)
(\hat{\Delta}(\hat{f} \,  \hat{g}))
\ff
& \equalref{lineamulti} &
\Omega \left(\hat{f}_{(1)} \, \hat{g}_{(1)} \right) 
\otimes 
\Omega \left(\hat{f}_{(2)} \, \hat{g}_{(2)} \right) 
\ff
& \equalref{Oscomp1} & 
f_{(1)} \star g_{(1)} \otimes f_{(2)} \star g_{(2)}
\, .  
\nn
\eae
\end{proof}
%%%%%%%%%%%%%%%%%%%%%%%%%%%%%%%%%%%%%%%%%%%%%%%%%%%%%%%%%%%%%%
\subsubsection{Translation-invariant integration}
\label{Tii}
%%%%%%%%%%%%%%%%%%%%%%%%%%%%%%%%%%%%%%%%%%%%%%%%%%%%%%%%%%%%%%
The above concept of translation for a function $\hat{f}(\hat{z})$ permits the definition of a translation-invariant integral 
$\bolo{\hat{f}(\hat{z})}$ as follows
\be 
\label{intdef}
\bolo{\hat{f}(\hat{z})} 
\equiv 
\int_{-\infty}^{\infty}  f(z)\, d^{n}z 
\equiv
\int f 
\, ,
\ee
assuming the \rhs{} exists.
Translation-invariance is defined and proved in
\begin{prop}
\label{transinvprop}
$\bolo{\cdot}$ is left- and right-invariant, 
\ie, it satisfies
\be
\label{tinvdef}
\bolo{\hat{f}} 1 =
\hat{f}_{(1)} \bolo{\hat{f}_{(2)}}
=
\bolo{\hat{f}_{(1)}} \hat{f}_{(2)}
\, .
\ee
\end{prop}
\begin{proof}
We prove left invariance, right invariance is proved analogously.
\bae
1 \bolo{\hat{f}}
\fe 
1\int f  
\ff 
\fe 
f_{(1)} \int f_{(2)} 
\ff
& \equalref{OmegaDelta} &
\Omega(\hat{f}_{(1)}) \int \Omega(\hat{f}_{(2)})
\ff
\fe
\Omega(\hat{f}_{(1)}) \bolo{\hat{f}_{(2)}}
\, , 
\nn
\eae
where, in the second equality, the translation invariance of the standard integral in $\mathbb{R}^N$ has been used.
But $\Omega^{-1}(1)=1$, so taking $\Omega^{-1}$ on both sides we obtain%
\footnote{%
Strictly speaking one should differentiate betwen the unit function $1_\calF$ and the unit $1_{\hat{\calF}}$.%
}
\be
1\bolo{\hat{f}}
=  
\hat{f}_{(1)}  \bolo{\hat{f}_{(2)}}
\, . \nn
\ee
\end{proof}
%%%%%%%%%%%%%%%%%%%%%%%%%%%%%%%%%%%%%%%%%%%%%%%%%%%%%%%%%%%%%%
\subsubsection{Quasi-cyclicity property of the integral}
\label{Tlp}
%%%%%%%%%%%%%%%%%%%%%%%%%%%%%%%%%%%%%%%%%%%%%%%%%%%%%%%%%%%%%%
Define the constants $c_A \equiv f_{AB}^{\phantom{AB}B}$. 
\begin{prop}
\label{trace}
$\bolo{\cdot}$ satisfies the {\em quasi-cyclicity} property
\be 
\label{tracedef}
\bolo{\hat{f}(\hat{z}) \hat{g}(\hat{z})}
=
\bolo{\hat{g}(\hat{z}+\lambda c) \hat{f}(\hat{z})}
\, .
\ee
\end{prop}
The proof is given in the appendix.
\begin{corollary}
If the adjoint representation of $\mathfrak{g}$ is traceless, $\bolo{\cdot}$ is cyclic.
\end{corollary}
This should be compared to Lemma 3.3 in~\cite{bieliavsky.bordemann.gutt.waldmann:2003a}.
\begin{corollary}
\label{cozc}
$\bolo{\cdot}$ is invariant under the $c$-number translation $\hat{z}_A \mapsto \hat{z}_A + \lambda c_A$. 
\end{corollary}
\begin{proof}
Put $\hat{f}=1$ in~(\ref{tracedef}).
\end{proof}
This last property deserves a comment. The invariance of 
$\bolo{\cdot}$ in the sense of~(\ref{tinvdef}) can be stated as follows: let $\hat{z}'$ be a second copy of the generators of $\mathfrak{g}$, satisfying the same Lie algebra, but commuting with the $\hat{z}$'s. Then 
\be 
\label{hatg1}
\hat{g}(\hat{z}+\hat{z}')
=
\hat{g}_{(1)}(\hat{z})\hat{g}_{(2)}(\hat{z}')
\, ,
\ee
and $\bolo{\hat{g}_{(1)}(\hat{z})} \hat{g}_{(2)}(\hat{z}')
=
\bolo{\hat{g}(\hat{z})} 1$. But for this property to hold, it is essential that the ``translation parameters'' $\hat{z}'$ commute with the $\hat{z}$'s {\em and satisfy the Lie algebra relations among themselves}. Yet, the $\lambda c$'s in Corollary~(\ref{cozc}), although they do commute with the $\hat{z}$'s, it would seem that they don't satisfy the Lie algebra relations, being numbers. Or do they? Well, the fact is they do, since $f_{AB}^{\phantom{AB}R}c_R=0$, as can be seen starting from the Jacobi identity for the structure constants. Thus, the $\lambda c$'s are, in some sense, both commuting and noncommuting entities. 
%%%%%%%%%%%%%%%%%%%%%%%%%%%%%%%%%%%%%%%%%%%%%%%%%%%%%%%%%%%%%%
\section{Examples}
\label{Examples}
%%%%%%%%%%%%%%%%%%%%%%%%%%%%%%%%%%%%%%%%%%%%%%%%%%%%%%%%%%%%%%
%%%%%%%%%%%%%%%%%%%%%%%%%%%%%%%%%%%%%%%%%%%%%%%%%%%%%%%%%%%%%%
\subsection{\mvb $1+1$  $\kappa$-Minkowski spacetime\mvn}
%%%%%%%%%%%%%%%%%%%%%%%%%%%%%%%%%%%%%%%%%%%%%%%%%%%%%%%%%%%%%%
The $1+1$ $\kappa$-Minkowski algebra is given by
\be
[\hat{t}, \hat{x}]= \lambda \, \hat{x} 
\ee
(we use $\lambda$, instead of the more usual $\kappa$, as deformation parameter). On the corresponding group manifold we use coordinates $\xi$, $\tau$ associated with the representation 
\ble{grep}
\rho(g)=\left(
\begin{array}{cc}
e^{\lambda \tau} & \xi
\\
0 & 1
\end{array}
\right)
\, ,
\ee
so that $\epsilon(\xi)=\epsilon(\tau)=0$.
The associated LIVFs are,
\be
\label{XTdef}
X = e^{\lambda \, \tau} x 
\, , 
\qquad
\qquad
T = t 
\, ,
\ee
where $x \equiv \partial_\xi$, $t \equiv \partial_\tau$. 
%%%%%%%%%%%%%%%%%%%%%%%%%%%%%%%%%%%%%%%%%%%%%%%%%%%%%%%%%%%%%%
\subsubsection{\mvb $(\Omega,\star)$ for $\kappa$-Minkowski\mvn}
%%%%%%%%%%%%%%%%%%%%%%%%%%%%%%%%%%%%%%%%%%%%%%%%%%%%%%%%%%%%%%
From the above, it is clear that
\be
\label{Omegaxmtn}
\Omega( \hat{x}^m \, \hat{t}^n)=\Omega(e^{m\lambda \tau} x^m \, t^n) =x^m t^n
\, .
\ee
It follows that, for a function $F(x,t)=\sum_i f_i(x)h_i(t)$, we have%
\footnote{%
This result coincides with the ``time to the right'' ordering various authors adopt, see, \eg,~\cite{Majid:1994cy}.
}
\ble{omegaiF}
\Omega^{-1}(F(x,t))=\sum_i f_i(\hat{x})h_i(\hat{t})
\, .
\ee 
In particular, putting $r \equiv e^{i(\alpha x +\beta t)}$, we find
\ble{omegaiexp}
\hat{r}=\Omega^{-1}(e^{i(\alpha x +\beta t)})
=
e^{i\alpha \hat{x}}e^{i\beta \hat{t}}
\, ,
\ee
so that
\bae
\label{omegaiexp2}
r^{(\bar{1})}  r^{(2)}
& \equalref{omegaiexp} &
e^{i\alpha X} e^{i\beta T}
\ff
& \equalref{XTdef} &
e^{i\alpha e^{\lambda \tau} x} e^{i\beta t}
\ff
\fe
e^{i\alpha(e^{\lambda \tau}-1)x} e^{i(\alpha x +\beta t)}
\ff
\fe
\sum_{n=0}^\infty
\frac{1}{n!}(i\alpha)^n (e^{\lambda \tau}-1)^n  x^n e^{i(\alpha x +\beta t)}
\ff
\fe
\sum_{n=0}^\infty
\frac{1}{n!} (e^{\lambda \tau}-1)^n  x^n \partial_x^n e^{i(\alpha x +\beta t)}
\, ,
\eae
a result that, using linearity, extends to arbitrary functions $h(x,t)$,
\ble{h1b2M}
h^{(\bar{1})}(\xi,\tau)   h^{(2)}(x,t)= 
\sum_{n=0}^\infty
\frac{1}{n!} (e^{\lambda \tau}-1)^n   x^n \partial_x^n h(x,t)
\, . 
\ee
Hard as it might look,~(\ref{h1b2M}) can be summed nicely in an exponential, due to the interesting formula
\be\label{remfor}
\sum_{n=0}^\infty
\frac{1}{n!} (e^a-1)^n z(z-1)\ldots (z-(n-1))=e^{az}
\, ,
\ee
the derivation of which is given in the appendix, and the fact that $x^n \partial_x^n=x\partial_x(x\partial_x-1)\ldots (x\partial_x-(n-1))$. Thus, we find
\be\label{h1b2Mexp}
h^{(\bar{1})}   h^{(2)}=e^{\lambda \tau x\partial_x}\triangleright h
\, ,
\ee
and the product $f \star h$ becomes, in standard notation,
\be\label{fstarh2}
f \star h= f e^{\lambda \overleftarrow{\partial_t}x \overrightarrow{\partial_x}} h
\, .
\ee
%%%%%%%%%%%%%%%%%%%%%%%%%%%%%%%%%%%%%%%%%%%%%%%%%%%%%%%%%%%%%%
\subsubsection{\mvb Invariant integration for $\kappa$-Minkowski\mvn}
%%%%%%%%%%%%%%%%%%%%%%%%%%%%%%%%%%%%%%%%%%%%%%%%%%%%%%%%%%%%%%
Under an integral sign, integration by parts allows us to write~(\ref{fstarh2}) in the forms
\ble{fstarhf}
\int
f\star h
=
\int
\left(
e^{-\lambda\partial_t \partial_x x} f
\right)
h
=
\int
f
\left(
e^{-\lambda\partial_t x\partial_x} 
h
\right)
\, ,
\ee
where $\int$ stands for $\iint_{-\infty}^\infty 
dx dt$.

The case where $h$ is a delta function is of particular interest, as it provides a measure of the fuzziness of the $\star$-product,
\ble{fdelta}
\bar{f}(x_0,t_0) \equiv \int f(x,t) \star \delta(x-x_0,t-t_0)
= 
\left( 
e^{-\lambda\partial_t\partial_x x} f
\right)(x_0,t_0)
\, .
\ee
For $(x_0,t_0)=(0,0)$, the above expression reduces to 
\ble{fdelta0}
\bar{f}(0,0)
= 
f(0,-\lambda)
\, ,
\ee
since $(\partial_x x)^n|_{x=0}=(1+x\partial_x)^n|_{x=0}=1$, due to the leftmost $x$ in all terms in the binomial expansion except the first one. Notice that, from~(\ref{fdelta0}), one may {\em not} deduce that $\bar{f}(x_0,t_0)=f(x_0,t_0-\lambda)$, as $\int f(x,t) \star \delta(x-x_0,t-t_0) \neq \int f(x+x_0,t+t_0) \star \delta(x,t)$, in general, as mentioned already in the previous section.

\begin{myexample}{Invariance and quasi-cyclicity of the integral in $\kappa$-Minkowski}{PsixkM}
The following formulas are easily established and will be useful in what follows
\be
g (\tg) \xg^m 
=
\xg^m g(\tg+m\lambda)
\, ,
\qquad
\qquad
e^{\alpha \tg} g(\xg) 
=
g(e^{\alpha \lambda} \xg) e^{\alpha \tg}
\, .
\ee
We also compute $c_x = 0$, $c_t = 1$. Consider the function $\hat{F} = e^{ \alpha \tg} e^{-\xg^2} e^{-\tg^2}$ with integral
\bae
\bolo{\hat{F}} \fe \int e^{-e^{2 \alpha \lambda}x^2} e^{ \alpha t} e^{-t^2} \ff
\fe \pi e^{\frac{\alpha^2}{4}}e^{- \alpha \lambda} \nn \, ,
\eae 
and coproduct
\be
\Delta(\hat{F}) = \sum_{a=0}^{\infty}\sum_{b=0}^{\infty}\frac{(-2)^a (-2)^b}{(a!) (b!)}\left( e^{ \alpha \tg} e^{-\xg^2}\xg^a e^{-\tg^2}\tg^b   \otimes e^{ \alpha \tg} e^{-\xg^2}\xg^a e^{-\tg^2}\tg^b   \right) \nn \, .
\ee
We compute
\bae
\bolo{\hat{F}_{(1)}} \hat{F}_{(2)} 
\fe
\sum_{a=0}^{\infty}\sum_{b=0}^{\infty}\frac{(-2)^a (-2)^b}{(a!) (b!)} \bolo{e^{ \alpha \tg} e^{-\xg^2}\xg^a e^{-\tg^2}\tg^b}   e^{ \alpha \tg} e^{-\xg^2}\xg^a e^{-\tg^2}\tg^b
\ff
\fe \pi e^{- \alpha \lambda} e^{ \alpha \tg}  e^{-\xg^2} \left( e^{-2 \xg \partial_s} \triangleright e^{\frac{1}{4}  s^2}\right)_{s=0}  e^{-\tg^2} \left(e^{-2 \tg \partial_\alpha} \triangleright e^{\frac{\alpha^2}{4}}\right)
\ff
\fe \pi e^{- \alpha \lambda} e^{\frac{\alpha^2}{4}} \nn \, ,
\eae
where we used
$\bolo{e^{ \alpha \tg} e^{-\xg^2}\xg^a e^{-\tg^2}\tg^b}
=
\pi  e^{- \alpha \lambda}  \left( \partial_s^a \triangleright e^{\frac{1}{4} s^2}\right)_{s=0} \left(\partial_\alpha^b \triangleright e^{\frac{\alpha^2}{4}}\right)$. Similarly
\bae
\bolo{\hat{F}_{(1)}} F_{(2)} 
\fe 
\pi e^{- \alpha \lambda}   e^{-e^{2 \alpha \lambda}x^2}  \left( e^{-2 e^{ \alpha \lambda} x \, \partial_s} \triangleright e^{\frac{1}{4}s^2}\right)_{s=0} e^{\alpha t} e^{-t^2} \left(e^{-2 t \partial_\alpha} \triangleright e^{\frac{\alpha^2}{4}}\right) 
\ff 
\fe \pi e^{- \alpha \lambda} e^{\frac{\alpha^2}{4}} \, . \nn
\eae
Therefore,
$
\bolo{\hat{F}}
=
\bolo{\hat{F}_{(1)}} \hat{F}_{(2)}
=
\bolo{ \hat{F}_{(1)}} F_{(2)}
$.
Next, we consider the integral
\bae 
\bolo{\hat{F}(\xg+\lambda c_x, \tg+\lambda c_t)} 
\fe 
\bolo{e^{ \alpha (\tg + \lambda)}e^{-\xg^2} e^{-(\tg+ \lambda)^2}}
\ff
\fe 
\int e^{-e^{2 \alpha \lambda}x^2} e^{ \alpha (t + \lambda)} e^{-(t + \lambda)^2}
\ff
\fe 
\int e^{-e^{2 \alpha \lambda}x^2} e^{ \alpha t} e^{-t^2}  
\ff
\fe 
\pi e^{- \alpha \lambda} e^{\frac{\alpha^2}{4}}\nn \, ,
\eae
so that $\bolo{\hat{F}( \hat{z})} = \bolo{\hat{F}( \hat{z}+\lambda c)}$.
Finally, we can write $\hat{F}= \hat{f} \, \hat{g}$ with $\hat{f}=e^{ \alpha \tg}$ and $\hat{g}=e^{-\xg^2}e^{-\tg^2}$, and calculate
\bae
\bolo{\hat{g}(\xg,\tg) \hat{f}(\xg-\lambda c_x, \tg-\lambda c_t)}
\fe 
\bolo{e^{-\xg^2} e^{-\tg^2} e^{ \alpha (\tg-\lambda)}} 
\ff
\fe  
e^{- \alpha \lambda} \int e^{-x^2}  e^{-t^2} e^{ \alpha t}  
\ff
\fe \pi e^{- \alpha \lambda} e^{\frac{\alpha^2}{4}}\nn \,,
\eae
showing that $\bolo{\hat{f}( \hat{z})\hat{g}( \hat{z})} = \bolo{\hat{g}( \hat{z}) \hat{f}( \hat{z}-\lambda c)}$. Notice that~(\ref{Omegaxmtn}) implies that, in this case, $\Omega\big(\hat{g}(\hat{z}+a)\big)=g(z+a)$, for $a$ a number, so that the star product version of~(\ref{tracedef}) is
\begin{equation}
\label{starvtr}
\int f(z) \star g(z)=\int g(z+\lambda c) \star f(z)
\, .
\end{equation}
\end{myexample}
%%%%%%%%%%%%%%%%%%%%%%%%%%%%%%%%%%%%%%%%%%%%%%%%%%%%%%%
%%%%%%%%%%%%%%%%%%%%%%%%%%%%%%%%%%%%%%%%%%%%%%%%%%%%%%%%%%%%%%
\subsubsection{\mvb $\delta_\lambda$ for $\kappa$-Minkowski\mvn}
%%%%%%%%%%%%%%%%%%%%%%%%%%%%%%%%%%%%%%%%%%%%%%%%%%%%%%%%%%%%%%
It is interesting to ponder whether there exists an object
$\delta_\lambda(x,t;x_0,t_0)$ such that
\ble{fdeltal}
\int f(x,t) \star \delta_\lambda(x,t;x_0,t_0)
= 
f(x_0,t_0)
\, .
\ee
Notice that we have carefully refrained from assuming that $\delta_\lambda$ only depends on differences of its arguments. 
Using the second form of~(\ref{fstarhf}) in~(\ref{fdeltal}) we get
\be\label{deltalM}
\delta_\lambda(x,t;x_0,t_0)
=
e^{\lambda \partial_t x\partial_x} \delta(x-x_0,t-t_0)
\, .
\ee

%%%%%%%%%%%%%%%%%%%%%%%%%%%%%%%%%%%%%%%%%%%%%%%%%%%%%%%%%%%%%%
\subsection{The Heisenberg algebra (``canonical'' NC plane)}
%%%%%%%%%%%%%%%%%%%%%%%%%%%%%%%%%%%%%%%%%%%%%%%%%%%%%%%%%%%%%%	
The Heisenberg algebra is a three-dimensional real Lie algebra with only nonzero commutator given by
\ble{Heisendef}
[\hat{x}, \hat{z}]= \lambda \, \hat{y} 
\, .
\ee
Notice that, since $\hat{y}$ is central, we are dealing essentially with the so-called ``canonical'' NC 2-plane.
On the Heisenberg group manifold we use coordinates $\xi$, $\psi$, $\zeta$, associated with the representation 
\ble{grepH}
\rho(g)=\left(
\begin{array}{ccc}
1 & \xi & \psi
\\
0 & 1 & \zeta
\\
0 & 0 & 1
\end{array}
\right)
\, ,
\ee
so that $\epsilon(\xi)=\epsilon(\psi)=\epsilon(\zeta)=0$.
The associated LIVFs are,
\ble{XYZdef}
X = x 
\, , 
\qquad
\qquad
Y = y 
\, , 
\qquad
\qquad
Z = z + \lambda \xi y
\, ,
\ee
where $x \equiv \partial_\xi$, $y \equiv \partial_\psi$, $z \equiv \partial_\zeta$. 
It is clear that
\ble{Ozyx}
\Omega( \hat{z}^m \, \hat{y}^n \, \hat{x}^k)=
\Omega
\big( 
(z+\lambda \xi y)^m \, y^n \, x^k) =z^m y ^n x^k
\, .
\ee
It follows that, for a function $F(x,y,z)=\sum_i f_i(z)h_i(y)s_i(x)$, we have
\ble{omegaiFH}
\Omega^{-1}(F(x,y,z))=\sum_i f_i(\hat{z})h_i(\hat{y})s_i(\hat{x})
\, .
\ee 
In particular, with $r \equiv e^{i(\alpha x +\beta y +\gamma z)}$, we find
\ble{omegaiexpH}
\hat{r}=\Omega^{-1}(e^{i(\alpha x +\beta y + \gamma z)})
=
e^{i\gamma \hat{z}}e^{i\beta \hat{y}}e^{i\alpha \hat{x}}
\, ,
\ee
so that
\bae
\label{omegaiexp2H}
r^{(\bar{1})}  r^{(2)}
& \equalref{omegaiexpH} &
e^{i\gamma (z+\lambda \xi y)} e^{i\beta y}e^{i\alpha x}
\ff
\fe
e^{i\gamma \lambda \xi y} e^{i(\alpha x +\beta y +\gamma z)}
\ff
\fe
\sum_{n=0}^\infty
\frac{\lambda^n}{n!} \xi^n  y^n 
\partial_z^n e^{i(\alpha x +\beta y +\gamma z)}
\, ,
\eae
a result that, by linearity, extends to arbitrary functions $h(x,y,z)$,
\ble{h1b2H}
h^{(\bar{1})}(\xi,\psi,\zeta) h^{(2)}(x,y,z)= 
\sum_{n=0}^\infty
\frac{\lambda^n}{n!} \xi^n  y^n 
\partial_z^n h(x,y,z)
\, . 
\ee
For the product $f \star h$ we find
\ble{fstarhH}
f \star h
=
(h^{(\bar{1})}(\partial) \triangleright f)h^{(2)}
\equalref{h1b2H}
\sum_{n=0}^\infty
\frac{\lambda^n}{n!} y^n (\partial_x^n f)
= 
f e^{\lambda y\overleftarrow{\partial_x} \overrightarrow{\partial_z}} h
\, ,
\ee
which gives, in particular, $x \star z=xz +\lambda y$, $z\star x=zx$, so that $[x,z]_\star=\lambda y$.

Under an integral sign we may use integration by parts to get
\ble{fstarh2H}
\int
f\star h
=
\int
\sum_{n=0}^\infty
\frac{1}{n!} (-\lambda )^n y^n \left(
\partial_x^n \partial_z^n f
\right) \, h
=
\int
\left( e^{-\lambda y \partial_x \partial_z}f \right) h
\, ,
\ee
where $\int$ stands for $\iiint_{-\infty}^\infty 
dx dy dz$.

For $h$ equal to a delta function we compute,
\ble{fdeltaH}
\bar{f}(x_0,y_0,z_0) 
\equiv 
\int f(x,y,z) \star \delta(x-x_0,y-y_0,z-z_0)
= 
\left( 
e^{-\lambda y \partial_x \partial_z}f
\right)(x_0,y_0,z_0)
\, .
\ee
For $y_0=0$, the above expression reduces to 
\ble{fdelta0H}
\bar{f}(x_0,0,z_0)
= 
f(x_0,0,z_0)
\, .
\ee

It is interesting to check the quasi-cyclicity property of the integral, Eq.~(\ref{tracedef}), in this case. We find that the constants $c_A$ of Proposition~\ref{Tlp} are zero, so that the integral is cyclic. Using~(\ref{fstargprop}), we compute
\ble{deltafH}
\int\delta(x-x_0,y-y_0,z-z_0) \star e^{-(x^2+y^2+z^2)/2}
=
\left(
e^{\lambda y z\partial_x-\frac{\lambda^2}{2} y^2 \partial_x^2}
\triangleright
e^{-(x^2+y^2+z^2)/2}
\right)(x_0,y_0,z_0)
\, ,
\ee
which, by cyclicity, should be equal to the \rhs{} of~(\ref{fdeltaH}), with $f$ equal to the gaussian.
Switching to quantum harmonic oscillator notation, $a_z=(\partial_z+z)/\sqrt{2}$, $a_z^\dagger=(-\partial_z+z)/\sqrt{2}$, and recognizing the gaussian as the ground state wavefunction, we find
\begin{align}
\label{cycliccomp1}
e^{\frac{\lambda}{\sqrt{2}} y \partial_x a_z^\dagger}
e^{-\frac{\lambda}{\sqrt{2}} y \partial_x a_z} \ket{0}
&=
e^{\frac{\lambda}{\sqrt{2}} y \partial_x a_z^\dagger}
e^{\frac{\lambda}{\sqrt{2}} y \partial_x a_z} \ket{0}
\ff
&=
e^{\lambda yz\partial_x - \frac{\lambda^2}{4} y^2\partial_x^2}
\ket{0}
\, ,
\end{align}
where, in the first equality, the sign of $a_z$ can be changed freely, as it annihilates the ground state, while the second follows from the Baker-Cambell-Hausdorff (BCH) formula.
On the other hand, 
\be 
e^{\frac{\lambda}{\sqrt{2}} y \partial_x a_z^\dagger}
e^{-\frac{\lambda}{\sqrt{2}} y \partial_x a_z}
=
e^{-\lambda y \partial_x \partial_z 
+\frac{\lambda^2}{4} y^2 \partial_x^2}
\, ,
\nn
\ee
by direct application of the BCH formula,
so that, comparing with~(\ref{cycliccomp1}), 
\be 
\label{cycliccomp2}
e^{\lambda yz\partial_x - \frac{\lambda^2}{4} y^2\partial_x^2} \ket{0}
=
e^{-\lambda y \partial_x \partial_z 
+\frac{\lambda^2}{4} y^2 \partial_x^2} \ket{0}
\, ,
\ee
from which the desired equality follows.

For $\delta_\lambda$ we find
\be\label{deltalH}
\delta_\lambda(x,y,z;x_0,y_0,z_0)=e^{\lambda y \partial_x \partial_z} \delta(x-x_0,y-y_0,z-z_0)
\, .
\ee
It is clear that in the above formulas the variable $y$ may be given
consistently a numerical value, as $\partial_y$ appears nowhere. Thus, by setting, \eg, $y=1$, one obtains $[x,z]_\star=\lambda$, \ie, the ``canonical'' NC 2-plane. One may effect this substitution, bearing in mind that $y$ can be reinstated in any expression by letting $\lambda \mapsto \lambda y$.
Eq.~(\ref{deltalH}) shows that, for fixed $y$, $\delta_\lambda(x,z;x_0,z_0)$ only depends on the differences $x-x_0$, $z-z_0$.
%%%%%%%%%%%%%%%%%%%%%%%%%%%%%%%%%%%%%%%%%%%%%%%%%%%%%%%%%%%%%%
%%%%%%%%%%%%%%%%%%%%%%%%%%%%%%%%%%%%%%%%%%%%%%%%%%%%%%%%%%%%%%
\section{Concluding Remarks}
\label{CR}
%%%%%%%%%%%%%%%%%%%%%%%%%%%%%%%%%%%%%%%%%%%%%%%%%%%%%%%%%%%%%%
%%%%%%%%%%%%%%%%%%%%%%%%%%%%%%%%%%%%%%%%%%%%%%%%%%%%%%%%%%%%%%
We have presented a star product for Lie-type noncommutative spacetimes, the corresponding $\Omega$-map, and a translation-invariant and quasi-cyclic integral, along with examples drawn from $\kappa$-Minkowski and canonical noncommutative spacetimes. 
A natural continuation of our work would be the development of applications to quantum field theory, as in, for example, Ref.{}~\cite{Aschieri:2006ye} --- this we defer to a future publication. Interesting technical questions that should also be addressed include:
\begin{itemize}
\item
The implementation in the formalism of symmetries, in particular, Poincar\'e covariance, perhaps appropriately deformed.
\item
The study of the transformation properties of the integral $\bolo{\cdot}$ under the above symmetry operations.
\item
Analytical issues associated to the pseudodifferential operators encountered.
\item
The structure of the $\Omega$-map, and the determination of the explicit form of the equivalence with Gutt's star product.
\end{itemize}
We plan on elucidating at least some of these in the near future.
%%%%%%%%%%%%%%%%%%%%%%%%%%%%%%%%%%%%%%%%%%%%%%%%%%%%%%%%%%%%%%
%%%%%%%%%%%%%%%%%%%%%%%%%%%%%%%%%%%%%%%%%%%%%%%%%%%%%%%%%%%%%%
\section*{Acknowledgments}
\label{Ack}
%%%%%%%%%%%%%%%%%%%%%%%%%%%%%%%%%%%%%%%%%%%%%%%%%%%%%%%%%%%%%%
%%%%%%%%%%%%%%%%%%%%%%%%%%%%%%%%%%%%%%%%%%%%%%%%%%%%%%%%%%%%%%
C.{} C.{} would like to thank his colleagues at NTUA and, in particular, George Zoupanos and Konstantinos Anagnostopoulos, for hospitality, and multifaceted support. The authors wish to thank Stefan Waldmann for bringing to their attention several relevant references. They also acknowledge partial financial support from DGAPA-UNAM projects IN 121306-3 and IN 108103-3.
%%%%%%%%%%%%%%%%%%%%%%%%%%%%%%%%%%%%%%%%%%%%%%%%%%%%%%%%%%%%%%
%%%%%%%%%%%%%%%%%%%%%%%%%%%%%%%%%%%%%%%%%%%%%%%%%%%%%%%%%%%%%%
\appendix
%%%%%%%%%%%%%%%%%%%%%%%%%%%%%%%%%%%%%%%%%%%%%%%%%%%%%%%%%%%%%%
%%%%%%%%%%%%%%%%%%%%%%%%%%%%%%%%%%%%%%%%%%%%%%%%%%%%%%%%%%%%%%
\section{Proofs}
\label{Proofs}
%%%%%%%%%%%%%%%%%%%%%%%%%%%%%%%%%%%%%%%%%%%%%%%%%%%%%%%%%%%%%%
%%%%%%%%%%%%%%%%%%%%%%%%%%%%%%%%%%%%%%%%%%%%%%%%%%%%%%%%%%%%%%
We give in this appendix various proofs that use in an essential way Hopf algebra techniques, familiarity with which is assumed. To establish the notation, we start by collecting some standard facts --- \cite{Maj:95} is an appropriate reference.

The action $g(\partial_z) \triangleright f(z)$ that appears in~(\ref{fstargprop}), is given by
\be\label{actder}
g(\partial_z) \triangleright f(z)
=
f_{(1)} \ip{g}{f_{(2)}}
=
f_{(2)} \ip{g}{f_{(1)}}
\, ,
\ee
where $\Delta(f)$ is the coproduct inferred from the primitive one of the $z$'s, Eq.~(\ref{Deltadef}), and 
\be\label{ipgf}
\ip{g(\partial_z)}{f(z)}=
\left. g(\partial_z) \triangleright f(z) \right\vert_{z=0}
\, .
\ee
The commutation relations between $g$ and $f$ above can be written as
\be\label{fgcrsD}
g(\partial_z)f(z)=
\left(g_{(1)}(\partial_z) \triangleright f(z)\right) g_{(2)}(\partial_z)
=
f_{(1)} \ip{g_{(1)}}{f_{(2)}} g_{(2)}
\, .
\ee

We use the standard Sweedler notation for the coproduct of $\psi(\zeta) \in \calA$, where $\calA$ is the algebra of functions on the group manifold,
\be 
\label{copgroup}
\Delta(\psi(\zeta)) = \psi_{(1)}(\zeta) \otimes \psi_{(2)}(\zeta)
\ee
(the use of the same symbol $\Delta$ as for the coproduct in $\calF$ should cause no confusion as the nature of the map can be inferred from that of its argument).
The semidirect product commutation relations with the LIVF $Z$ are given by 
\be
\label{Zpsicrs}
Z \psi=\psi_{(1)} \ip{Z_{(1)}}{\psi_{(2)}} Z_{(2)}
\, ,
\ee
where $\Delta(Z)$ is the primitive one, and the inner product $\ip{\cdot}{\cdot}\colon \mathfrak{g} \otimes \calA \rightarrow \mathbb{R}$ is defined by
\be
\label{ipdef}
\ip{Z_A}{\psi(\zeta)}= \left. \frac{\partial \psi(\zeta)}{\partial \zeta^A}\right\vert_{\zeta=0}
\, .
\ee
Eq.~(\ref{Zpsicrs}) can be extended from $\mathfrak{g} \otimes \calA$ to $\hat{\calF} \otimes \calA$ by multiplicativity and linearity of $\hat{\Delta}$, and the relation
\be\label{ipextend}
\ip{XY}{\psi}=\ip{X}{\psi_{(1)}}\ip{Y}{\psi_{(2)}}
\, ,
\ee
with $X$, $Y$ in $\hat{\calF}$.

The {\em left adjoint action} $\triadj \colon \hat{\calF} \otimes \hat{\calF} \rightarrow \hat{\calF}$, defined on $\mathfrak{g} \otimes \mathfrak{g}$ by
\be 
\label{adjointacdef}
Z_A \triadj Z_B = f_{AB}^{\phantom{AB}C} Z_C
\, ,
\ee
is extended to $\hat{\calF} \otimes \hat{\calF}$ by linearity in both arguments and the relations
\be\label{DAextend}
(XY) \triadj V=X \triadj(Y \triadj V)
\, ,
\qquad
\qquad
X \triadj (YV)=(X_{(1)} \triadj Y)(X_{(2)} \triadj V)
\, .
\ee
The {\em right adjoint coaction}
\be 
\label{adjointcodef}
\Delta_\calA \colon \hat{\calF} \rightarrow \hat{\calF} \otimes \calA
\, ,
\qquad
X \mapsto \Delta_\calA(X) \equiv X^{(1)} \otimes X^{(2)'}
\, ,
\ee
is dual to the adjoint action $\triadj$, 
in the sense that
\be\label{duality}
X \triadj Y = Y^{(1)} \ip{X}{Y^{(2)'}}
\, ,
\ee
and satisfies dualized versions of~(\ref{DAextend}),
\be\label{Ddualv}
(\Delta_\calA \otimes \id)\Delta_\calA
=
(\id \otimes \Delta)\Delta_\calA
\, ,
\qquad
\qquad
\Delta_\calA(XY)=\Delta_\calA(X) \Delta_\calA(Y)
\, .
\ee 
Putting $\Delta_\calA(Z_B) \equiv Z_C \otimes M^C_{\phantom{C}B}$ one obtains 
\be\label{TMip}
\ip{Z_A}{M^C_{\phantom{R}B}}=f_{AB}^{\phantom{AB}C}
\, .
\ee
It also holds 
\be\label{Mepsilon}
\epsilon(M^A_{\phantom{A}B})
\equiv
\ip{1}{M^A_{\phantom{A}B}}
=
\delta^A_{\phantom{A}B}
\, .
\ee
The commutation relations in $\hat{\calF}$ are given by
\be\label{hFcrs}
XY
=
(X_{(1)} \triadj Y) X_{(2)} 
=
Y^{(1)} \ip{X_{(1)}}{Y^{(2)'}} X_{(2)}
\, .
\ee
\textbf{Proposition~\ref{OmegaDeltaprop}.}
\begin{proof}
The proof goes by induction. Since $\omega(Z_A)=\partial_A$, it is trivial to check that the property holds for $\hat{z}_A$. Next, assume that it is true for $\hat{f}(\hat{z})$, a polynomial of order $n$ in the $\hat{z}_A$, 
and consider the polynomial of order $n+1$, $\hat{f}(\hat{z}) \hat{z}_A$.

\noindent On one hand,
\bae 
\Omega\left(\hat{f}(\hat{z})\hat{z}_A\right)
& \equalref{Omegadef1} &
\omega\left(\hat{f}(Z) Z_A\right)
\ff
& \equalref{omegaprop1} &
\omega\left(f(z)Z_A\right)
\ff
& \equalref{Zgs} &
\omega\left(f(z)\psi_A^{\phantom{A}B} z_B\right)
\ff
& \equalref{fgcrsD} &
\omega\left(
{\psi_A^{\phantom{A}B}}_{(1)}
\ip{f_{(1)}}{{\psi_A^{\phantom{A}B}}_{(2)}} f_{(2)} z_B
\right)
\ff
& \equalref{omegadef} &
\epsilon({\psi_A^{\phantom{A}B}}_{(1)})
\ip{f_{(1)}}{{\psi_A^{\phantom{A}B}}_{(2)}} f_{(2)} z_B
\ff
\fe 
\ip{f_{(1)}}{\psi_A^{\phantom{A}B}} f_{(2)} z_B
\nn
\, ,
\eae
so that
\bae
\Delta\left(\Omega(\hat{f}(\hat{z}) \hat{z}_A)\right)
\fe
\ip{f_{(1)}}{\psi_A^{\phantom{A}B}}
(f_{(2)} {z_B}_{(1)} \otimes f_{(3)} {z_B}_{(2)})
\ff
& \equalref{Deltadef} &
\ip{f_{(1)}}{\psi_A^{\phantom{A}B}}
(f_{(2)} z_B \otimes f_{(3)} 
+
f_{(2)} \otimes f_{(3)} z_B)
\label{DOproof1}
\, .
\eae
\noindent On the other hand,
\bae
\big( \Omega \otimes \Omega \big) 
\hat{\Delta}\left(
\hat{f}(\hat{z}) \hat{z}_A
\right)
& \equalref{Omegadef1} &
(\omega \otimes \omega)
\left(
\hat{f}_{(1)}(Z) Z_A \otimes \hat{f}_{(2)}(Z) +
\hat{f}_{(1)}(Z) \otimes \hat{f}_{(2)}(Z) Z_A
\right)
\ff
& \equalref{omegaprop1} &
(\omega \otimes \omega)
\left(
\omega\big(\hat{f}_{(1)}(Z)\big) Z_A \otimes \hat{f}_{(2)}(Z) +
\hat{f}_{(1)}(Z) \otimes \omega\big(\hat{f}_{(2)}(Z)\big) Z_A
\right)
\ff
\fe 
(\omega \otimes \omega)
\left(
f_{(1)}(z) Z_A \otimes f_{(2)}(z)+
f_{(1)}(z)  \otimes f_{(2)}(z) Z_A
\right)
\ff
& \equalref{actder} &
\ip{f_{(1)}}{\psi_A^{\phantom{A}B}} f_{(2)} z_B \otimes f_{(3)}
+ f_{(1)} \otimes \ip{f_{(2)}}{\psi_A^{\phantom{A}B}}f_{(3)} z_B
\, ,
\label{DOproof2}
\eae
where the third equality used the induction hypothesis. The proposition now follows by comparison of~(\ref{DOproof1}), (\ref{DOproof2}), and the coassociativity of the coproduct.
\end{proof}
\textbf{Proposition~\ref{trace}.}
\begin{proof}
We first show that
\be\label{intfgK}
\bolo{\hat{f}\hat{g}}
=
\bolo{K(\hat{g}) \hat{f}}
\, ,
\ee
where
\be\label{Kdef}
K(\hat{g}) \equiv
\ip{
S(\hat{g}^{(1)}_{\phantom{(1)}(1)})
}{\hat{g}^{(2)'}}
\hat{g}^{(1)}_{\phantom{(1)}(2)}
\, ,
\ee
with $S$ denoting the antipode and $\hat{g}^{(1)}_{\phantom{(1)}(1)} \otimes \hat{g}^{(1)}_{\phantom{(1)}(2)}\otimes \hat{g}^{(2)'} \equiv (\Delta \otimes \operatorname{id})\Delta_\calA(\hat{g})$. Indeed,
\bae
\bolo{\hat{f}\hat{g}}
& \equalref{hFcrs} &
\ip{\hat{f}_{(1)}}{\hat{g}^{(2)'}}
\bolo{\hat{g}^{(1)} \hat{f}_{(2)} } 
\ff
\fe
\ip{
\hat{f}_{(1)} S(\hat{f}_{(2)}) S(\hat{g}^{(1)}_{\phantom{(1)}(1)})}{\hat{g}^{(2)'}}
\bolo{\hat{g}^{(1)}_{\phantom{(1)}(2)} \hat{f}_{(3)}}
\ff
\fe
\ip{
S(\hat{g}^{(1)}_{\phantom{(1)}(1)})}{\hat{g}^{(2)'}}
\bolo{
\hat{g}^{(1)}_{\phantom{(1)}(2)} \hat{f}}
\, ,
\eae
where in the second line, we used the invariance of the integral to multiply $\hat{f}_{(1)}$ from the right by $S(1)$.
By repeated application of~(\ref{intfgK}), the multiplicativity of $K$ is established%
\footnote{%
Direct proof of the multiplicativity of $K$ from its definition, Eq.~(\ref{Kdef}), is left as a non-trivial exercise for the reader ({\em Hint:} the Jacobi identity is needed).%
}, so, given its obvious linearity, it suffices to compute it for the generators $Z_A$. We find
\be\label{DDZ}
(\Delta \otimes \operatorname{id})\Delta_\calA(Z_A)
=
\Delta(Z_B) \otimes M^{B}_{\phantom{B}A}
=
Z_B \otimes 1 \otimes M^{B}_{\phantom{B}A}
+
1 \otimes Z_B \otimes M^{B}_{\phantom{B}A}
\, ,
\ee
so that 
\bae
K(Z_A)
\fe
\ip{-Z_B}{M^{B}_{\phantom{B}A}}1+
\ip{1}{M^{B}_{\phantom{B}A}}Z_B
\ff
&  \equalref{Mepsilon} &
Z_A -\ip{Z_B}{M^{B}_{\phantom{B}A}}
\ff
& \equalref{TMip} &
Z_A - f_{BA}^{\phantom{BA}B}
\ff
\fe
Z_A +c_A
\, ,
\eae
implying $K(\hat{g}(Z))=\hat{g}(Z+c)$.
\end{proof}
\textbf{Eq.~(\ref{remfor}).}
\begin{proof}
The Hopf algebras $\calR$, of functions on the real line, with coordinate $x$, and $\calP$, of functions of $\partial_x$, both equipped with the standard Hopf structure, are dually paired, via the inner product
\be\label{APipdef}
\ip{\partial_x}{x}=1
\, ,
\ee
extended to $\calP \otimes \calR$ in the standard fashion,
\be\label{pxsf}
\ip{g(\partial_x)}{f(x)}
=
\left( g(\partial_x) \triangleright f(x)
\right)_{x=0}
\, .
\ee 
It is easily established that $\{x^n\}$, $\{\frac{1}{n!}\partial_x^n\}$, are dual bases, and so are $\{x^{(n)}\}$, $\{\frac{1}{n!}D_x^n\}$, where $x^{(n)}=x(x-1)\ldots (x-(n-1))$, and $D_x \triangleright f(x)=f(x+1)-f(x)=(e^{\partial_x}-1) \triangleright f$. The desired identity then follows from the invariance of the canonical element $C$ under linear change of bases,
\bae
C
\fe
\sum_{n=0}^\infty
\frac{1}{n!}x^n \otimes \partial_x^n
\ff
\fe
e^{x \otimes \partial_x}
\ff
\fe
\sum_{n=0}^\infty
\frac{1}{n!}x^{(n)} \otimes D_x^n
\ff
\fe
\sum_{n=0}^\infty
\frac{1}{n!}x^{(n)} \otimes (e^{\partial_x}-1)^n
\, ,
\eae
with the identification $x \otimes 1 \mapsto z$, $1 \otimes \partial_x \mapsto a$.
\end{proof}
%%%%%%%%%%%%%%%%%%%%%%%%%%%%%%%%%%%%%%%%%%%%%%%%%%%%%%%%%%%%%%
%%%%%%%%%%%%%%%%%%%%%%%%%%%%%%%%%%%%%%%%%%%%%%%%%%%%%%%%%%%%%%
\section{Further Examples}
\label{FE}
%%%%%%%%%%%%%%%%%%%%%%%%%%%%%%%%%%%%%%%%%%%%%%%%%%%%%%%%%%%%%%
%%%%%%%%%%%%%%%%%%%%%%%%%%%%%%%%%%%%%%%%%%%%%%%%%%%%%%%%%%%%%%
We present two more examples dealing with integration in $\kappa$-Minkowski spacetime. 
%%%%%%%%%%%%%%%%%%%%%%%%%%%%%%%%%%%%%%%%%%%%%%%%%%%%%%%%%%%%%%
\begin{myexample}{}{appex3}
%%%%%%%%%%%%%%%%%%%%%%%%
In the context of Ex.~\ref{PsixkM}, let $\hat{G} = e^{-\xg^2} e^{-\tg^2} \xg^m$. Then, some straightforward algebra gives
\bae
\Omega(\hat{G}) 
\fe 
e^{-x^2} x^m e^{-(t+ m \lambda)^2}
\ff
\bolo{\hat{G}} 
\fe 
\pi \left(\partial_r^m \triangleright e^{\frac{r^2}{4}}\right)_{r=0} 
\ff
\Delta(\hat{G}) 
\fe 
\sum_{k=0}^{m}\sum_{a=0}^{\infty}\sum_{b=0}^{\infty} 
\binom{m}{k}
\frac{(-2)^a (-2)^b}{a! b!}  
\left( e^{-\xg^2}\xg^a  e^{-\tg^2} \tg^b \xg^k  \otimes e^{-\xg^2} \xg^a   e^{-\tg^2} \tg^b \xg^{m-k} \right) 
\ff
\Omega(e^{-\xg^2}\xg^a  e^{-\tg^2} \tg^b \xg^k) 
\fe 
e^{-x^2}x^{a+k}  e^{-(t+  k\lambda)^2}(t+  k \lambda)^b 
\ff
\bolo{e^{-\xg^2}\xg^a  e^{-\tg^2} \tg^b \xg^k} 
\fe  
\pi \left(\partial_s^{a+k} \triangleright  e^{\frac{s^2}{4}}\right)_{s=0} \left(\partial_r^b \triangleright e^{\frac{r^2}{4}}\right)_{r=0} 
\ff
\bolo{\hat{G}_{(1)}} \hat{G}_{(2)}
\fe 
\pi e^{-\xg^2} \left( \left( e^{-2 \xg \partial_r}  e^{-\tg^2} \left( e^{-2 \tg \partial_s} \triangleright e^{\frac{s^2}{4}}\right)_{s=0} (\partial_r+\xg)^m \right) \triangleright e^{\frac{r^2}{4}}\right)_{r=0} 
\ff
\fe 
\pi \left( (\partial_r+\xg)^m \triangleright e^{\frac{r^2}{4}-r \xg}\right)_{r=0} 
\ff
\bolo{\hat{G}_{(1)}} G_{(2)} \fe
\pi \sum_{k=0}^{m}
\binom{m}{k}
 e^{-(t+  k \lambda)^2} \left(e^{-2 (t+  k \lambda)\partial_s}
\triangleright e^{\frac{s^2}{4}}\right)_{s=0} e^{-x^2} \left(
x^{m-k} \partial_r^k  e^{-2 x \partial_r} \triangleright  
e^{\frac{r^2}{4}}\right)_{r=0} 
\ff
\fe
\pi \left((\partial_r+x)^m \triangleright e^{\frac{r^2}{4}-r x}\right)_{r=0} \nn \, .
\eae
However, one easily computes
\bae
\left( (\partial_r+c)^m \triangleright e^{\frac{r^2}{4}-r c}\right)_{r=0} \fe
\sum_{k=0}^{m}
\binom{m}{k}
c^{m-k} \left(\partial_r^k \triangleright e^{\frac{r^2}{4}} e^{-r c}\right)_{r=0}\ff
\fe 
\sum_{k=0}^{m}\sum_{l=0}^{k}
\binom{m}{k}
\binom{k}{l}
c^{m-k} \left(\partial_r^l \triangleright e^{\frac{r^2}{4}} \right)_{r=0} \left(\partial_r^{k-l} \triangleright e^{-r c} \right)_{r=0} \ff
\fe 
\sum_{k=0}^{m}\sum_{l=0}^{k}
\binom{m}{k}
\binom{k}{l}
c^{m-k}(-c)^{k-l} \left(\partial_r^l \triangleright e^{\frac{r^2}{4}} \right)_{r=0}  \ff
\fe 
\sum_{k=0}^{m}
\binom{m}{k}
c^{m-k} \left((\partial_r-c)^k \triangleright e^{\frac{r^2}{4}} \right)_{r=0}  \ff
\fe 
\left(\partial_r^m \triangleright e^{\frac{r^2}{4}}\right)_{r=0} \nn \, ,
\eae
so that
\be 
\bolo{\hat{G}}
=
\bolo{\hat{G}_{(1)}} \hat{G}_{(2)}
=
\bolo{\hat{G}_{(1)}} G_{(2)}
= 
\pi \left(\partial_r^m \triangleright e^{\frac{r^2}{4}}\right)_{r=0} \, . \nn
\ee
Consider also the integral
\bae
\bolo{\hat{G}(\xg+\lambda c_x, \tg+\lambda c_t)} 
\fe 
\bolo{e^{-\xg^2}  e^{-(\tg + \lambda)^2} \xg^m} 
\ff
\fe 
\int e^{-x^2} x^m e^{-(t+ \lambda + m \lambda)^2}  
\ff
\fe 
\pi \left(\partial_r^m \triangleright e^{\frac{r^2}{4}}\right)_{r=0}\nn \, ,
\eae
showing that $\bolo{\hat{G}( \hat{z})} 
= 
\bolo{\hat{G}( \hat{z}+\lambda c)}$.
Finally, if we write $\hat{G}=\hat{g} \, \hat{h}$ with $\hat{g}=e^{-\xg^2} e^{-\tg^2}$ and $\hat{h}=\xg^m$, and compute
\bae
\bolo{\hat{h}(\xg,\tg) \hat{g}(\xg-\lambda c_x, \tg-\lambda c_t)} 
\fe 
\bolo{\xg^m e^{-\xg^2} e^{-(\tg- \lambda)^2}} 
\ff
\fe  
\int x^m e^{-x^2}  e^{-(t- \lambda)^2} 
\ff
\fe \pi \left(\partial_r^m \triangleright e^{\frac{r^2}{4}}\right)_{r=0} \nn \, ,
\eae
we obtain, $\bolo{\hat{g}( \hat{z})\hat{h}( \hat{z})} 
= 
\bolo{ \hat{h}( \hat{z}) \hat{g}( \hat{z}-\lambda c)}$.
\end{myexample}
%%%%%%%%%%%%%%%%%%%%%%%%%%%%%%%%%%%%%%%%%%%%%%%%%%%%%%%%%%%%%%%%%%%%
\begin{myexample}{}{appex4}
%%%%%%%%%%%%%%%%%%%%%%%%%%
Consider the function $\hat{H}(\xg, \tg)=e^{-\xg^2} e^{-\tg^2} \xg^m e^{ \alpha \tg}$, with
\bae
\bolo{\hat{H}(\xg, \tg)} 
\fe 
\bolo{e^{-\xg^2} e^{-\tg^2} \xg^m e^{ \alpha \tg}} 
\ff
\fe 
\int e^{-x^2} x^m e^{-(t+ m \lambda)^2} e^{ \alpha t}  
\ff
\fe 
\left(\int x^m e^{-x^2}\right)  \sqrt{\pi} e^{\frac{\alpha^2}{2}} e^{- \alpha m \lambda}  \nn \, .
\eae
Now write $\hat{H}= \hat{g} \, \hat{k}$, with $\hat{g}=e^{-\xg^2}e^{-\tg^2} $ and $\hat{k}=\xg^m e^{ \alpha \tg}$, and compute the integral
\bae
\bolo{\hat{k}(\xg,\tg) \hat{g}(\xg-\lambda c_x, \tg-\lambda c_t)}
\fe 
\bolo{\xg^m e^{ \alpha \tg } e^{-\xg^2}e^{-(\tg-\lambda)^2}}
\ff
\fe  
\int x^m  e^{-e^{2 \alpha \lambda}x^2} e^{ \alpha t} e^{-(t-\lambda)^2}    
\ff
\fe 
e^{- \alpha \lambda(m+1) }\int x^m e^{-x^2} e^{ \alpha t}e^{-(t-\lambda) ^2}    
\ff
\fe 
\left(\int x^m e^{-x^2}\right)  \sqrt{\pi} e^{\frac{\alpha^2}{2}} e^{- \alpha m \lambda}  \nn \, .
\eae
Once again, $\bolo{\hat{g}( \hat{z})\hat{k}( \hat{z})} 
= 
\bolo{\hat{k}( \hat{z}) \hat{g}( \hat{z}-\lambda c)}$.
Similarly
\bae
\bolo{\hat{k}(\xg+\lambda c_x,\tg+\lambda c_t) \hat{g}(\xg, \tg)} 
\fe 
\bolo{\xg^m e^{ \alpha (\tg+\lambda) } e^{-\xg^2}e^{-\tg^2}} 
\ff
\fe  
\int x^m  e^{-e^{2 \alpha \lambda}x^2} e^{ \alpha (t+\lambda)} e^{-t^2}     
\ff
\fe 
e^{- \alpha m \lambda } \int x^m e^{-x^2} e^{\alpha t}e^{-t^2}  
\ff
\fe 
\left(\int x^m e^{-x^2}\right)  \sqrt{\pi} e^{\frac{\alpha^2}{2}} e^{- \alpha m \lambda}  
\nn \, ,
\eae
so $\bolo{\hat{g}( \hat{z})\hat{k}( \hat{z})} 
= 
\bolo{\hat{k}( \hat{z}+\lambda c) \hat{g}( \hat{z})}$.
\end{myexample}
%%%%%%%%%%%%%%%%%%%%%%%%%%%%%%%%%%%%%%%%%%%%%%%%%%%%%%%%%%%%%%
%%%%%%%%%%%%%%%%%%%%%%%%%%%%%%%%%%%%%%%%%%%%%%%%%%%%%%%%%%%%%%
%%%%%%%%%%%%%%%%%%%%%%%%%%%%%%%%%%%%%%%%%%%%%%%%%%%%%%%%%%%%%%
%\begin{comment}
%\bibliographystyle{plain}
%\bibliography{/windows/D/chryss/papers/strings,%
%/windows/D/chryss/papers/dqarticle,%
%/windows/D/chryss/papers/dqbook,%
%/windows/D/chryss/papers/dqthesis,%
%/windows/D/chryss/papers/dqprocentry,%
%/windows/D/chryss/papers/dqproceeding}
%\end{comment}
%%%%%%%%%%%%%%%%%%%%%%%%%%%%%%%%%%%%%%%%%%%%%%%%%%%%%%%%%%%%%%
\providecommand{\href}[2]{#2}
\begingroup\raggedright
\endgroup
%%%%%%%%%%%%%%%%%%%%%%%%%%%%%%%%%%%%%%%%%%%%%%%%%%%%%%%%%%%%%%
\end{document}